\documentclass[twocolumn,amsmath,amssymb,10pt,pra]{revtex4-1}
\usepackage[utf8]{inputenc}
\usepackage[T1]{fontenc}
\usepackage{amsmath}
\usepackage{amsfonts}
\usepackage{graphicx}
\usepackage{yfonts}
\usepackage{color}
\usepackage[normalem]{ulem}
\usepackage{amsthm}
\usepackage{bm}
\usepackage{bbm}
\usepackage{mathtools}
\usepackage{array}
\usepackage{enumitem}

\usepackage{color}

 \newcommand{\der}{\mathrm{d}}

\def\<{\langle}
\def\>{\rangle}

\newcommand{\Tr}{\mathrm{Tr}}
\def\oper{{\mathchoice{\rm 1\mskip-4mu l}{\rm 1\mskip-4mu l}
{\rm 1\mskip-4.5mu l}{\rm 1\mskip-5mu l}}}
\usepackage{caption}
\newtheorem{Lemma}{Lemma}

\newtheorem{Proposition}{Proposition}

\begin{document}

\title{\bf Geometry on the manifold of Gaussian quantum channels}
\author{Katarzyna Siudzi{\'n}ska}
\affiliation{Institute of Physics, Faculty of Physics, Astronomy and Informatics \\  Nicolaus Copernicus University, Grudzi\k{a}dzka 5/7, 87--100 Toru{\'n}, Poland}
\author{Kimmo Luoma}
\author{Walter T. Strunz}
\affiliation{Institut f\"{u}r Theoretische Physik, Technische Universit\"{a}t Dresden, D-01062, Dresden, Germany}

\begin{abstract}
In the space of quantum channels, we establish the geometry that 
allows us to make statistical predictions about relative volumes of entanglement breaking channels among all the Gaussian quantum channels. The underlying metric is constructed using the Choi-Jamio{\l}kowski isomorphism between 
the continuous-variable Gaussian states and channels. This construction involves the Hilbert-Schmidt distance in quantum state space. The volume element of the one-mode Gaussian channels can be expressed in terms of local symplectic invariants. We analytically compute the relative volumes of the one-mode Gaussian entanglement breaking and incompatibility breaking channels. Finally, we show that, when given the purities of the Choi-Jamio{\l}kowski state of the channel, one can determine whether or not such channel is incompatibility breaking.
\end{abstract}

\maketitle

\section{Introduction}\label{sec:introduction}
Quantum correlations, like non-locality~\cite{RevModPhys.86.419}, 
steering~\cite{2019arXiv190306663U}, and entanglement~\cite{RevModPhys.81.865},
are very often used as the key resources in quantum information tasks,
such as quantum state discrimination~\cite{Bae_2015} and key distribution~\cite{Razavi:19}. 
Unfortunately, quantum systems are never perfectly isolated from the 
external influences, which leads to a
harmful dissipation and decoherence that ultimately destroy quantum 
correlations~\cite{Yu598}. These dynamical processes for the open systems
can be described using quantum channels, which are completely positive,
trace-preserving maps from quantum states to quantum
states~\cite{breuer2002theory}. Exceptionally detrimental kinds of open
quantum system dynamics are described by the so-called {\it
entanglement breaking} (EB)~\cite{doi:10.1142/S0129055X03001709} and
{\it incompatibility breaking} (ICB)~\cite{Heinosaari_2015} quantum
channels. The former maps any entangled input state to a separable
output. The latter, on the other hand, has the corresponding dual map
that, when applied to a pair of incompatible (not jointly measurable)
observables, maps them to a pair of compatible (jointly measurable)
observables~\cite{Heinosaari_2016}. Therefore, it is of general
interest to determine the likelihood to  encounter such channels in the 
space of all channels, measured in terms of ratios of their corresponding 
volumes.

In this article, we focus on continuous variable (CV) systems; i.e.,
the systems described with the help of the canonical position and momentum
operators.  Laboratories equipped with linear optical elements and
photodetectors can routinely prepare, manipulate, and perform quantum
measurements on the states of such
systems~\cite{Ferraro2005,Olivares2012}. An important special class of
states for continuous variable systems is the set of Gaussian states,
characterized by the Gaussian Wigner
function~\cite{Holevo,Adesso2014,Olivares2012}. Closely associated to
them is the set of Gaussian quantum channels, mapping any Gaussian
state to a Gaussian state. These sets of states and channels form 
basic building blocks for current experiments on photonic systems in
the field of quantum
information~\cite{Braunstein2005,Weedbrook2012}. The main reasons for
their appeal is the fact that the Gaussianity-preserving unitary
operations can be implemented in linear optics, and Gaussian systems
are relatively easy to handle mathematically.

In this article, we study 
the geometry of the manifold of Gaussian quantum channels. 
We provide a rigorous route to investigate how likely it is, among all Gaussian quantum channels,
to encounter a channel that is either entanglement breaking \cite{Holevo2008} or
incompatibility breaking~\cite{doi:10.1063/1.4928044}. So far, the
investigations on the information geometry in the Gaussian domain have
been focused on the geometry of the state space
\cite{Link2015,Felice2017} and the typical properties of quantum
correlations~\cite{Dahlsten2014,Serafini2007,Sohr_2018}. In~\cite{Monras2010},
first steps are taken to study the geometry of the Gaussian
quantum channels. The main hindrance for further development have been
the technical difficulties that are encountered when one tries to
formulate the Choi-Jamio{\l}kowski (CJ)
isomorphism~\cite{CHOI1975285,JAMIOLKOWSKI1972275} for continuous
variable systems \cite{Holevo2011,Holevo2008}.

Recently, new results have shed some light on how to formulate the
Choi-Jamio{\l}kowski correspondence between Gaussian states and
channels in such a way that divergence problems do not
occur~\cite{Kiukas2017,PhysRevA.97.022339}. In this article, we use
the approach developed in \cite{Kiukas2017} for the
Choi-Jamio{\l}kowski isomorphism in combination with the results on
the geometry of Gaussian states in \cite{Sohr_2018} to investigate the
geometry of Gaussian quantum channels. In particular, we report the
likelihood of encountering a one-mode entanglement or incompatibility
breaking channel among all the one-mode Gaussian channels. It should be
noted that such results, in general, depend on the choice of the metric. 
Here, as the metric on the space of channels is defined with the help
of the Choi-Jamio{\l}kowski isomorphism, it will also depend on the reference state of that isomorphism.

The rest of the article is organized as follows.  In Section
\ref{sec:gaussian-states} we provide a quick review on the main
properties of the Gaussian states. Then, in Section
\ref{sec:gauss-chann-choi}, we introduce the notion of the Gaussian
channels. We present the generalization of the Choi-Jamio{\l}kowski
correspondence that is valid for the continuous variable systems. In
Section \ref{sec:hilb-schm-dist}, we construct the Hilbert-Schmidt
line and volume elements on the manifold of Gaussian states with a
fixed marginal. This section contains our first main result: the
volume element of the one-mode Gaussian channels. We use this result
in Section \ref{sec:entangl-incomp-break}, where we compute the
relative volumes for the entanglement breaking and incompatibility
breaking channels. In the last section, we present conclusions,
open questions, and a final outlook.

\section{Gaussian states}\label{sec:gaussian-states}
Consider an $n$-particle continuous variable system with the
corresponding Hilbert space
$\mathcal{H}:=\bigotimes\limits_{k=1}^nL^2(\mathbb{R})$. The canonical
operators acting on $\mathcal{H}$ can be arranged to create a vector
$R:=(q_1,p_1,\ldots,q_n,p_n)^T$ with $q_k:=a_k^\dagger+a_k$ and $p_k:=
i (a_k^\dagger-a_k)$. The creation and annihilation operators satisfy
the bosonic commutation relations $[a_i,a_j^\dagger]=\delta_{ij}\mathbb{I}$,
$[a_i,a_j]=0$, which induce the following relation for the vector
components,
\begin{align*} 
  [R_i,R_j^\dagger]=2 i
  \Omega_{ij},\quad\Omega:=\bigoplus_{k=1}^n\omega,\quad
  \omega:=\begin{pmatrix} 0 & 1 \\ -1 & 0 \end{pmatrix}.
\end{align*} 
In the above equations, $\Omega$ is the symplectic form. Now, for every
quantum state $\rho$, let us introduce the characteristic function
$\chi(\xi):=\Tr[D(\xi)\rho]$, where
$\xi:=(\xi_1^{(1)},\xi_1^{(2)},\ldots,\xi_n^{(1)},\xi_n^{(2)})^T\in\mathbb{R}^{2n}$
are the phase space coordinates and 
\begin{align*} 
D(\xi)&:=e^{ i R^T\Omega\xi},
  \\ D(\xi)D(\xi^\prime)&=e^{-i\xi^T\Omega\xi'}D(\xi')D(\xi),
\end{align*}
are the displacement (or Weyl) operators~\cite{HILLERY1984121,doi:10.1063/1.4928044}.
By definition, a Gaussian state is a quantum
state whose characteristic function $\chi(\xi)$ is a Gaussian
function \cite{Weedbrook2012,doi:10.1063/1.4928044}. 
We write it as 
\begin{equation*} 
\chi(\xi)=\exp\left[-\frac 12
\xi^T\Omega\Sigma\Omega^T\xi+ i \ell^T\Omega\xi\right],
\end{equation*} 
where $\ell_k:=\Tr(\rho R_k)$ is the displacement vector
and $\Sigma_{ij}:=\frac{1}{2}( \Tr[\rho(R_iR_j+R_jR_i)]-\ell_i\ell_j)$ is
the covariance matrix of the underlying Gaussian quantum state $\rho=\rho(\Sigma,\ell)$. 
Then, the state can be expressed as 
\begin{equation*}
\rho(\Sigma,\ell):=\int_{\mathbb{R}^{2n}}\frac{\der^{2n}\xi}{\pi^n}\chi(\xi)D(-\xi).
\end{equation*}
Note that  $\Sigma$ is a covariance matrix of a
Gaussian state if and only if
\begin{equation}\label{CP} 
\Sigma+i\Omega\geq 0,
\end{equation}
due to the canonical commutation relations. We would like to stress
that, despite the apparent similarities between the Gaussian states in
the classical and quantum domain, the quantum case is fundamentally
different due to eq.~(\ref{CP}). The classical Gaussian states
(probability densities) can become arbitrarily narrow approaching the
Dirac $\delta$ function
in a limiting sense, whereas eq.~(\ref{CP}) sets the minimal admissible width
for the Gaussian quantum states compatible with Heisenberg's uncertainty relation.

\section{Gaussian channels and the Choi-Jamio{\l}kowski isomorphism}\label{sec:gauss-chann-choi}

The Gaussian quantum channels $\Lambda:\mathcal{H}_A\to\mathcal{H}_B$ 
are completely positive, trace-preserving maps that transform Gaussian quantum states into Gaussian quantum states. 
The action of a Gaussian channel leads to a dual map 
on the displacement operators \cite{Holevo,Holevo2007,holevo2011probabilistic}
\begin{equation*}\label{channel}
\Lambda^\ast[D(\xi)]=D(M\xi)\exp\left[-\frac 12 \xi^TN\xi+i c^T\xi\right],
\end{equation*}
with the matrices $M\in\mathbb{M}_{2n_A\times 2n_B}(\mathbb{R})$,
$N=N^T\in\mathbb{M}_{2n_B\times 2n_B}(\mathbb{R})$, and the
$2n_A$-dimensional vector $c$. Therefore, each channel is completely
characterized by a triple $(M,N,c)$. The action of a Gaussian channel
on a Gaussian state $\rho(\Sigma,\ell)$ is then efficiently expressed in
terms of the covariance matrix and the displacement vector,
\begin{equation*}
\Sigma\mapsto M^T\Sigma M+N,\qquad \ell\mapsto M^T\ell+c.
\end{equation*}
The complete positivity condition
\begin{equation*}
N-iM^T\Omega M+i\Omega\geq 0
\end{equation*}
follows directly from eq. (\ref{CP}).

In order to leverage the known results on the geometry of Gaussian
states~\cite{Link2015,Sohr_2018}, we use the Choi-Jamio{\l}kowski
isomorphism to express the Gaussian channels in terms of the state
parameters. Let us recall Lemma 4 from Kiukas
et. al.~\cite{Kiukas2017}.

\begin{Lemma}\label{L1}
There exists a one-to-one correspondence between bipartite Gaussian
states $\rho_{AB}$ with a common marginal $\sigma=\Tr_A\rho_{AB}$, a
covariance matrix $\Sigma_\sigma$ of a full symplectic rank, and a
displacement $\ell_\sigma$, and Gaussian channels
$\Lambda:\mathcal{H}_B\to\mathcal{H}_A$, such that
\begin{equation*}
\rho_{AB}=(\Lambda\otimes\oper_B)(\rho_\Omega).
\end{equation*}
The Gaussian state $\rho_\Omega$ is characterized by the following covariance matrix and displacement,
\begin{equation*}
\Sigma_\Omega:=\begin{pmatrix}
\Sigma_\sigma & S_\sigma^TZ_\sigma S_\sigma \\ S_\sigma^TZ_\sigma S_\sigma & \Sigma_\sigma
\end{pmatrix},\qquad \ell_\Omega:=\ell_\sigma\oplus\ell_\sigma.
\end{equation*}
In the above definition,
\begin{align*}
&\Sigma_\sigma=:S_\sigma^TD_\sigma S_\sigma,\\ &D_\sigma:=\mathrm{diag}(\nu_{\sigma,1},\nu_{\sigma,1},\ldots,\nu_{\sigma,N},\nu_{\sigma,N})\\
&Z_\sigma:=\bigoplus_{k=1}^N\sigma_3\sqrt{\nu_{\sigma,k}^2-1},
\end{align*}
with $S_\sigma$ being the symplectic matrix ($S_\sigma^T\Omega S_\sigma=\Omega$) diagonalizing $\Sigma_\sigma$.
\end{Lemma}

The correspondence between the Gaussian channel $\Lambda(M,N,c)$ and
the Gaussian Choi-Jamio{\l}kowski (CJ) state $\rho_{AB}(\Sigma,\ell)$
is given as follows~\cite{Kiukas2017},
\begin{equation}\label{corr}
\begin{cases}
\Sigma&=\begin{pmatrix}
\Sigma_A & \Gamma^T \\ \Gamma & \Sigma_\sigma
\end{pmatrix},\\
\ell&=\ell_A\oplus\ell_\sigma,
\end{cases}
\qquad
\begin{cases}
M&=(S_\sigma^TZ_\sigma S_\sigma)^{-1}\Gamma,\\
N&=\Sigma_A-M^T\Sigma_\sigma M,\\
c&=\ell_A-M^T\ell_\sigma.
\end{cases}
\end{equation}

Obviously, the isomorphism in Lemma~\ref{L1} depends on the 
choice of $\sigma$, and so will the metric properties of the channel space. 

\section{Geometry of Gaussian states and channels}\label{sec:hilb-schm-dist}
In order to discuss the geometric properties in the space of quantum channels, we need to 
define a metric in terms of a line element.
For finite-dimensional systems, there exists the unique
unitarily invariant line element induced by the Fubini-Study
metric \cite{Bengtsson2006}. In the infinite-dimensional case,
however, there are many possible non-equivalent choices for the
metric~\cite{Sohr_2018}. We base our calculations on the
Hilbert-Schmidt distance defined by $\der s^2=\Tr(\der \rho^2)$. On
the manifold of the Gaussian states, it takes
the following form,
\begin{equation*}\label{dhs}
\begin{split}
\der s^2=&\Tr[\rho(\Sigma+\der\Sigma,\ell+\der\ell)]^2
+\Tr[\rho(\Sigma,\ell)]^2\\&
-2\Tr[\rho(\Sigma,\ell)\rho(\Sigma+\der\Sigma,\ell+\der\ell)].
\end{split}
\end{equation*}
For more details considering the computation of the line element, see Appendix~\ref{sec:hilbert-schmidt-line}. The final result is
\begin{equation}\label{line}
\begin{split}
\der s^2=
\frac{1}{16\sqrt{\det\Sigma}}\Big\{2\Tr[\Sigma^{-1}\der\Sigma]^2+[\Tr(\Sigma^{-1}\der\Sigma)]^2&\\
+8\der\ell^T\Sigma^{-1}\der\ell\Big\}.&
\end{split}
\end{equation}
Hence, the line element $\der s^2$ and the volume element $\der V$ can be written as
\begin{equation*}\label{line2}
\der s^2=\begin{pmatrix}\der\mathbf{\Sigma}^T & \der\ell^T\end{pmatrix}
\begin{pmatrix}G & 0 \\ 0 & g\end{pmatrix}
\begin{pmatrix}\der\mathbf{\Sigma} \\ \der\ell\end{pmatrix},
\end{equation*}
\begin{equation}\label{volume}
\der V=\sqrt{\det G}\sqrt{\det g}\prod_{i=1}^{4n^2}\der\mathbf{\Sigma}_i
\prod_{j=1}^{2n}\der\ell_j,
\end{equation}
where $\der\mathbf{\Sigma}=\mathrm{vec}(\der\Sigma)$ is the matrix vectorization. 
For $\ell=0$, eq. (\ref{line}) is in correspondence with the results obtained in \cite{Sohr_2018}. Note that
$\der\Sigma$ and $\der\ell$ are not coupled, and therefore a non-zero
displacement produces the multiplicative factor
\begin{equation*}
\sqrt{\det g}
=\left[\det\left(\frac{\Sigma^{-1}}{2\sqrt{\det\Sigma}}\right)\right]^{-1}
=2^{-n}(\det\Sigma)^{-\frac{n+1}{2}}
\end{equation*}
in the volume element.

Let us consider the one-mode Gaussian channel
$\Lambda(M,N,c)$. According to eqs. (\ref{corr}), the corresponding
two-mode Gaussian CJ state is given by
\begin{equation}\label{Sigma2}
\begin{split}
\Sigma=&
\begin{pmatrix}
N+M^T\Sigma_\sigma M & M^TS_\sigma^TZ_\sigma S_\sigma \\
S_\sigma^TZ_\sigma S_\sigma M & \Sigma_\sigma
\end{pmatrix},\\
\ell=&(c+M^T\ell_\sigma)\oplus\ell_\sigma.
\end{split}
\end{equation}
A non-zero displacement vector $\ell$ corresponds to local unitary contributions 
of the channel. When considering the effect of the channel on the 
non-local correlations, we can -- without the loss of generality -- set $\ell=0$, as we do 
for the rest of the article.

Note that any two-mode Gaussian covariance matrix 
can be expressed in the standard form $\Sigma=SWS^T$,  
where $S$ is a local symplectic transformation~\cite{Sohr_2018} and
\begin{align*}
  W = \begin{pmatrix}
  \nu_A  & 0 &\gamma_+ & 0 \\
  0 &\nu_A & 0 & \gamma_-\\
  \gamma_+ & 0 & \nu_\sigma & 0 \\
  0 & \gamma_- & 0 & \nu_\sigma
\end{pmatrix}.
\end{align*}
Here, $\nu_A$ and $\nu_\sigma$ are the symplectic eigenvalues of the 
marginal states, and $\gamma_\pm$ describe the correlations between the 
two modes.
Following the method presented in~\cite{Sohr_2018}, we compute the
Hilbert-Schmidt volume element for the Gaussian states with the
covariance matrix given by eq. (\ref{Sigma2}) (for more details, see
Appendix~\ref{sec:hilb-schm-volume}). This way, we obtain
\begin{equation*}
\der V=\sqrt{\det G}\der\nu_A\der\gamma_+\der\gamma_-\der\theta\der m(S_A),
\end{equation*}
where $\der m(S_A)$ is the measure of the non-compact symplectic group $Sp(2)$ and
\begin{equation*}\label{detG}
\sqrt{\det G}=\frac{\nu_A^2\nu_\sigma^3(\gamma_+^2-\gamma_-^2)}
{32\sqrt{2}(\gamma_+^2-\nu_A\nu_\sigma)^{17/4}
(\gamma_-^2-\nu_A\nu_\sigma)^{17/4}}.
\end{equation*}
It turns out that quantum correlations in the two-mode Gaussian 
states are most conveniently analyzed in the purity-seralian coordinates~\cite{Adesso2004}.
For the readers' convenience, we recall the definitions,
\begin{equation}\label{pur}
\begin{split}
\mu_{A/\sigma}&:=\frac{1}{\sqrt{\det\Sigma_{A/\sigma}}}=\frac{1}{\nu_{A/\sigma}},\\
\mu&:=\frac{1}{\sqrt{\det\Sigma}}=\frac{1}{\sqrt{(\gamma_+^2-\nu_A\nu_\sigma)(\gamma_-^2-\nu_A\nu_\sigma)}},\\
\Delta&:=\det\Sigma_A+\det\Sigma_\sigma+2\det\Gamma=\nu_A^2+\nu_\sigma^2+2\gamma_+\gamma_-.
\end{split}
\end{equation}
As it is apparent from their definitions, these four new coordinates are local symplectic invariants.
The inverse relations read
\begin{equation*}\label{cpm}
\nu_{A/\sigma}=\frac{1}{\mu_{A/\sigma}},\qquad
\gamma_\pm=\frac{\sqrt{\mu_A\mu_\sigma}}{4}(\epsilon_+\pm\epsilon_-),
\end{equation*}
where
\begin{equation*}
\epsilon_\pm:=\sqrt{\left(\Delta-\frac{(\mu_A\pm\mu_\sigma)^2}{\mu_A^2\mu_\sigma^2}\right)^2
-\frac{4}{\mu^2}}.
\end{equation*}
Finally, we obtain the formula
\begin{equation*}
\der V=\frac{\mu^{11/2}}{64\sqrt{2}\mu_A^3\mu_\sigma^2}
\der\mu_A\der\mu\der\Delta\der\theta\der m(S_A).
\end{equation*}
The range of coordinates is determined by $0\leq\mu_{A/\sigma}\leq 1$, $\epsilon_\pm^2\geq 0$, 
and by the complete positivity condition (\ref{CP}), which is equivalent to
\begin{equation*}
1+\frac{1}{\mu^2}-\Delta\geq 0.
\end{equation*}
Combining all these requirements results in the following conditions
for the two-mode Gaussian CJ states that correspond to legitimate
one-mode Gaussian channels \cite{Adesso2004},
\begin{equation}\label{range}
\begin{split}
&0\leq\mu_{A/\sigma}\leq 1,\qquad \mu_A\mu_\sigma\leq\mu\leq\frac{\mu_A\mu_\sigma}{\mu_A\mu_\sigma+|\mu_A-\mu_\sigma|},\\
&\frac{2}{\mu}+\frac{(\mu_A-\mu_\sigma)^2}{\mu_A^2\mu_\sigma^2}\leq\Delta\leq
\min\left\{-\frac{2}{\mu}+\frac{(\mu_A+\mu_\sigma)^2}{\mu_A^2\mu_\sigma^2},1+\frac{1}{\mu^2}\right\}.
\end{split}
\end{equation}
Now, we want to express conditions~(\ref{range}) for the complete positivity of the channel in terms of the channel parameters directly. 
Interestingly, it turns out that these conditions depend only on the determinants of $M$ and $N$.
\begin{Proposition}\label{prop:CPMN}
Any one-mode Gaussian map $\Lambda$ characterized by $(M,N,c)$ is completely positive if and only if
\begin{equation}\label{eq:CP_cond}
\det N\geq (\det M-1)^2.
\end{equation}
\end{Proposition}

\begin{proof}
For one-mode Gaussian channels, the complete positivity condition in eq. (\ref{CP}) is equivalent to \cite{Simon2000}
\begin{equation*}
\det(\Sigma+i\Omega)\geq 0.
\end{equation*}
Note that $\Sigma+i\Omega$ is a $4\times 4$ matrix, so its determinant can be calculated using the property
\begin{equation}\label{CP2}
\det(\Sigma+i\Omega)=\det D\det F.
\end{equation}
In the above formula, $D:=\Sigma_\sigma+i\omega$ and $F:=\Sigma_A+i\omega-\Gamma D^{-1}\Gamma^T$ for $\Sigma$ being a block matrix from eq. (\ref{corr}). Simple calculations on $2\times 2$ matrices show that $\det D=\nu_\sigma^2-1$ and
\begin{equation*}
F=N+i\omega(1-\det M),
\end{equation*}
where we implemented the formulas for $\Sigma_A$ and $\Gamma$ given on the r.h.s. of eq. (\ref{corr}). As $\det D\geq 0$, condition (\ref{CP2}) can be rewritten into
\begin{equation*}
\det F=\det N-(1-\det M)^2\geq 0.
\end{equation*}
\end{proof}

\section{Entanglement and incompatibility breaking channels}\label{sec:entangl-incomp-break}

Let us consider a special class of quantum channels, for which $\rho_{AB}=(\Lambda\otimes\oper_B)(\rho_\ast)$ is
separable for any (even entangled) state $\rho_\ast$. These are known
as the {\it entanglement breaking channels} and can always be written
in the Holevo form \cite{doi:10.1142/S0129055X03001709}
\begin{equation}\label{EBC}
\Lambda[\rho]=\sum_k\omega_k\Tr(F_k\rho),
\end{equation}
where $\omega_k$ are quantum states, and $F_k$ form a POVM (positive operator-valued measure). For finite-dimensional quantum systems, it is straightforward to show that $\Lambda$ is entanglement breaking if and only if $\rho_{AB}$ is separable for $\rho_\ast$ being a maximally entangled state. This
notion can be extended to infinite-dimensional systems if one replaces
the maximally entangled state with $\rho_\Omega$ from Lemma \ref{L1}.

\begin{Lemma}\label{L2}
A Gaussian channel $\Lambda$ is entanglement breaking if and only if
\begin{equation*}
\rho_{AB}=(\Lambda\otimes\oper_B)(\rho_\Omega),
\end{equation*}
with a marginal $\sigma=\Tr_A\rho_{AB}$, is separable.
\end{Lemma}

\begin{proof}
Note that if $\Lambda$ is entanglement breaking, then trivially
$\rho_{AB}$ is separable. Now, assume that $\rho_{AB}$ is separable;
i.e., $\rho_{AB}=\sum_kp_k\omega_k\otimes\beta_k$ with density
operators $\omega_k$, $\beta_k$ and a probability distribution
$p_k$. Then, we show that
\begin{equation}\label{TRA}
\Tr_A[\rho_{AB}(A\otimes\mathbb{I}_B)]=\sum_kp_k\beta_k\Tr(\omega_kA).
\end{equation}
Recall that for an arbitrary $\rho_{AB}$ with $\sigma=\Tr_A\rho_{AB}$, one has \cite{Kiukas2017}
\begin{equation*}\label{TRA2}
\sqrt{\sigma}\Lambda^\ast[A]\sqrt{\sigma}=\Tr_A[\rho_{AB}(A\otimes\mathbb{I}_B)]^T,
\end{equation*}
where $\Lambda^\ast$ is a map dual to the Gaussian channel $\Lambda$. Therefore, eq. (\ref{TRA}) is equivalent to
\begin{equation*}
\sqrt{\sigma}\Lambda^\ast[A]\sqrt{\sigma}=\sum_kp_k\beta_k\Tr(\omega_kA).
\end{equation*}
One cannot simply invert $\sqrt{\sigma}$, as the inverse of a full-rank state is unbounded.
However, 
$\sigma^{-1/2}p_k\beta_k\sigma^{-1/2}$ extends to a bounded operator $F_k$ for which $\sqrt{\sigma}F_k\sqrt{\sigma}=p_k\beta_k$,
because $\|\sqrt{F_k}\psi\|^2\leq\sum_k\<\psi|F_k\psi\>=\|\psi\|^2$ and $\|F_k\|=\|\sqrt{F_k}\|^2$. 
Now, we can see that $\Lambda^\ast$ is dual to the entanglement breaking channel of the form (\ref{EBC}). 
Indeed, the $F_k$ define a POVM, as
\begin{equation*}
\sum_k\sqrt{\sigma}F_k\sqrt{\sigma}=\sum_kp_k\beta_k=\Tr_A\rho_{AB}=\sigma.
\end{equation*}
\end{proof}

It was shown \cite{Simon2000} that for two-mode Gaussian states, the
Peres-Horodecki criterion \cite{PhysRevLett.77.1413,HORODECKI1997333}
is necessary and sufficient for separability. Namely, $\rho_{AB}$ is
separable if and only if
\begin{equation}\label{PPT}
\det(\Sigma_{PPT}+i\Omega)\geq 0,
\end{equation}
where $\Sigma_{PPT}=\Theta\Sigma\Theta$ is the covariance matrix of
the partially transposed state, and
$\Theta=\mathrm{diag}(-1,1,1,1)$. In the seralian-purity coordinates,
condition (\ref{PPT}) reads
\begin{equation}\label{range2}
1+\frac{1}{\mu^2}+\Delta-\frac{2}{\mu_A^2}-\frac{2}{\mu_\sigma^2}\geq 0.
\end{equation}

\begin{Proposition}\label{prop:EBC}
Any one-mode Gaussian channel $\Lambda$ characterized by $(M,N,c)$ is entanglement breaking if and only if
\begin{equation}\label{PPTMN}
\det N\geq (\det M+1)^2.
\end{equation}
\end{Proposition}

The proof is analogous to the proof of Proposition \ref{prop:CPMN}.

Now, consider another class of quantum channels, for which
$\rho_{AB}=(\Lambda\otimes\oper_B)(\rho_\ast)$ is non-steerable
for any choice of $\rho_\ast$. These channels are the so-called {\it
incompatibility breaking channels}~\cite{doi:10.1063/1.4928044}. It is known that a one-mode
Gaussian channel is incompatibility breaking if and only if
\begin{equation*}
\Sigma+i(0\oplus\omega)\geq 0,
\end{equation*}
or, in terms of the purities,
\begin{equation}\label{range3}
\mu\leq\mu_A,
\end{equation}
which are, in fact, just the conditions for the steerability of the CJ state.
\begin{Proposition}\label{prop:ICB}
Any one-mode Gaussian channel $\Lambda$ characterized by $(M,N,c)$ is incompatibility breaking if and only if
\begin{equation}\label{ICBMN}
\det N\geq\det M^2.
\end{equation}
\end{Proposition}

In Fig.~\ref{fig:domains}, one can see the graphical representation of the 
conditions from Propositions 1--3 for the one-mode Gaussian
channels. The complete positivity domain from ineq.~(\ref{eq:CP_cond}) is
gray, the entanglement breaking domain from ineq.~(\ref{PPTMN}) is
double-hatched, and the incompatibility breaking domain from ineq.~(\ref{ICBMN}) is single-hatched. We see that the EB domain is
contained within the ICB domain, and both of these domains are
contained in the CP domain. Note that the inequalities presented in 
Propositions \ref{prop:CPMN}--\ref{prop:ICB} are known \cite{Adesso2004,Holevo2008,doi:10.1063/1.4928044}.
Here, however, we were able to bring them to a unified concise form involving 
the simple channel parameters $\det M$ and $\det N$.
\begin{figure}[ht!]
  \includegraphics[width=0.45\textwidth]{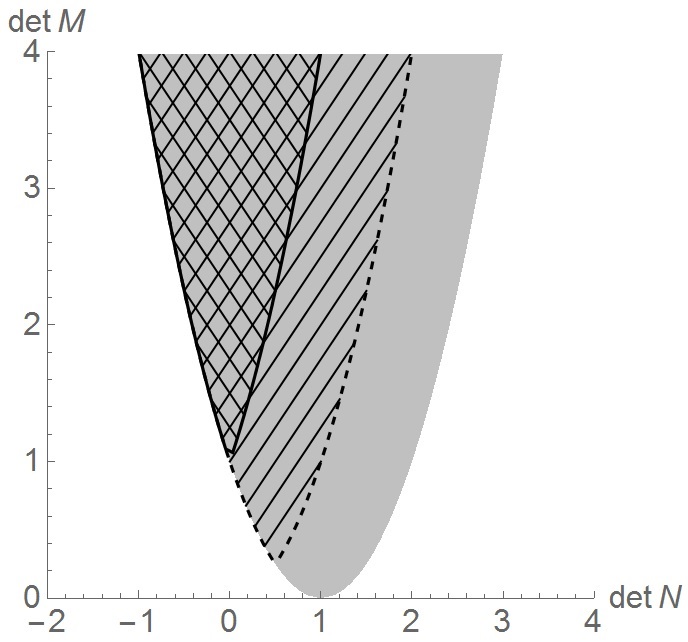}
\caption{\label{fig:domains}The range of $\det M$ and $\det N$ for which the complete positivity (gray), 
entanglement breaking (double-hatched), 
or incompatibility breaking (single-hatched) conditions are satisfied.}
\end{figure}

\section{Relative volumes}\label{sec:volume-ratios}
We analyze the geometry of the one-mode Gaussian channels by
considering the manifold of the corresponding Gaussian CJ states. In
order to do this, we make use of the local symplectic decomposition of
the covariance matrix $\Sigma$. Recall that the local symplectic group
$Sp(2)$ is non-compact~\cite{Adesso2007}, which means that the volume
of two-mode Gaussian states, and hence the one-mode Gaussian channels,
is not finite. The non-compactness emerges due to the possibility of
unbounded squeezing. Regardless, we can compute the relative volumes
of the quantities that are invariant with respect to the local
symplectic transformations, such as entanglement. This quantity can be
seen as the likelihood of encountering special classes of channels
among all one-mode Gaussian channels.

To calculate the total volume of all one-mode Gaussian channels, we
need to integrate the volume element in eq. (\ref{volume}) over the
range of parameters determined by ineq. (\ref{range}). Namely, one has
\begin{equation*}
V_{GC}=C\iiint_{\mathcal{CP}}\frac{\mu^{11/2}}{64\sqrt{2}\mu_A^3\mu_\sigma^2}
\der\mu_A\der\mu\der\Delta,
\end{equation*}
where $\mathcal{CP}$ is the region given by conditions (\ref{range}). We use the shorthand notation 
\begin{equation*}
C=\int_{\mathcal{M}}\der m(S_A) \int_0^{2\pi}\der\theta
\end{equation*}
for the divergent part of the integral.
Analogously, one obtains the volume of all entanglement breaking
channels $V_{EBC}$ and incompatibility breaking channels $V_{ICBC}$;
\begin{equation*}
V_{EBC}=C\iiint_{\mathcal{SEP}}\frac{\mu^{11/2}}{64\sqrt{2}\mu_A^3\mu_\sigma^2}
\der\mu_A\der\mu\der\Delta,
\end{equation*}
\begin{equation*}
V_{ICBC}=C\iiint_{\mathcal{NS}}\frac{\mu^{11/2}}{64\sqrt{2}\mu_A^3\mu_\sigma^2}
\der\mu_A\der\mu\der\Delta.
\end{equation*}
The regions of integration $\mathcal{SEP}$, $\mathcal{NS}$ are given
by conditions (\ref{range}, \ref{range2}) and (\ref{range},
\ref{range3}), respectively. Each of the above integrals can be solved
analytically. The results are
\begin{equation*}
V_{GC}=C\frac{4+\mu_\sigma^{9/2}(9\mu_\sigma^2-13)}{18018\sqrt{2}\mu_\sigma^3},
\end{equation*}
\begin{equation*}
V_{EBC}=C\frac{\sqrt{\mu_\sigma}(1-\mu_\sigma)^2(11+9\mu_\sigma)}{18018\sqrt{2}},
\end{equation*}
\begin{equation*}
V_{ICBC}=C\frac{\sqrt{\mu_\sigma}\left(-13\mu_\sigma+9\mu_\sigma^3
-\frac{8\sqrt{2}(-11+7\mu_\sigma)}{(1+\mu_\sigma)^{7/2}}\right)}{18018\sqrt{2}}.
\end{equation*}
It is easy to see that the divergent part $C$ drops out when one considers a
ratio of volumes. Interestingly, such ratio still depends on the choice of
$\rho_\Omega$ in the Choi-Jamio{\l}kowski isomorphism through the
marginal purity $\mu_\sigma$. The relative volumes of the entanglement 
and incompatibility breaking channels are presented in Fig.~\ref{fig:volumes} as
dashed and solid lines, respectively. Both curves grow
monotonically as functions of $\mu_\sigma$. Two points, $\mu_\sigma=0$
and $\mu_\sigma=1$, have to be excluded from our considerations, even
though the curves seem to behave well. The former point would
correspond to the maximally entangled $\rho_\Omega$, which is not a
trace-class operator. The latter point does not satisfy the conditions in
Lemma \ref{L1}.

\begin{figure}[ht!]
  \includegraphics[width=0.45\textwidth]{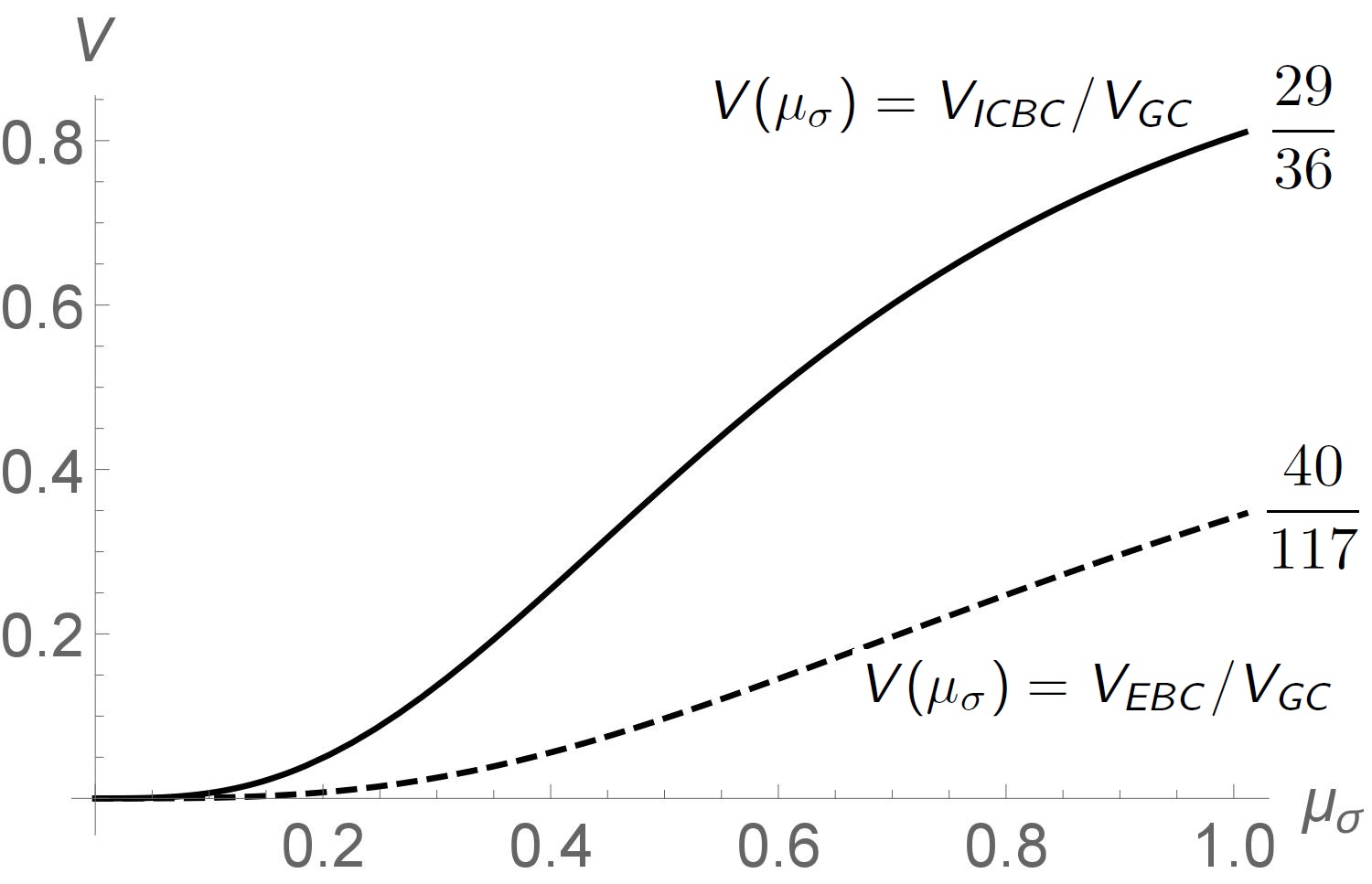}
  \caption{\label{fig:volumes}The relative volume of the entanglement breaking
    (dashed line) and incompatibility breaking (solid line) one-mode Gaussian channels 
    as a function of the marginal purity of the CJ state.}
\end{figure}

Now, assume that our knowledge about the two-mode Gaussian CJ state is
limited to the values of total $\mu$ and marginal $\mu_A$,
$\mu_\sigma$ purities. It turns out that even without knowing the value of the
seralian $\Delta$, we can usually tell whether a given one-mode
Gaussian channel is entanglement breaking or not. 
For fixed $\mu_A$ and $\mu_\sigma$, we see that 
if the total purity belongs to the range
\begin{equation*}
\mu_A\mu_\sigma\leq\mu\leq\frac{\mu_A\mu_\sigma}{\mu_A+\mu_\sigma-\mu_A\mu_\sigma}
\end{equation*}
or
\begin{equation*}
\frac{\mu_A\mu_\sigma}{\sqrt{\mu_A^2+\mu_\sigma^2-\mu_A^2\mu_\sigma^2}}
\leq\mu\leq\frac{\mu_A\mu_\sigma}{\mu_A\mu_\sigma+|\mu_A-\mu_\sigma|},
\end{equation*}
then the associated two-mode Gaussian CJ states are separable or
entangled, respectively \cite{Adesso2004}. There also exists the so-called {\it
coexistence region}, which corresponds to
\begin{equation*}
\frac{\mu_A\mu_\sigma}{\mu_A+\mu_\sigma-\mu_A\mu_\sigma}\leq\mu\leq
\frac{\mu_A\mu_\sigma}{\sqrt{\mu_A^2+\mu_\sigma^2-\mu_A^2\mu_\sigma^2}}.
\end{equation*}
For such values of $\mu_{A/\sigma}$ and $\mu$, it is impossible to
distinguish between the separable and entangled states without the
full knowledge about the system. 
Having the expressions for the 
volumes at hand, we can even compute the probability of finding an entangled 
CJ state in the coexistence region (see the density plots in Fig. 3).

In Fig.~\ref{fig:cond_prob} we plot the separability (double-hatched), coexistence
(single-hatched), and entanglement (shaded) regions for the Gaussian
CJ states as functions of $\mu$ and $\mu_A$ for fixed $\mu_\sigma$.
The unhatched white region is unphysical. Interestingly, knowing the
value of seralian is not necessary to determine the steerability of
Gaussian states. The shading indicates the relative conditional
volume of entangled Gaussian CJ states. This is computed as a ratio of
the volume of entangled states with respect to the volume of all
states for fixed purities. 

Discussing these figures in terms of the one-mode Gaussian quantum channels,
we can simply read off the incompatibility and entanglement breaking regions
in the parameter space. When $\mu\leq\mu_A$, the channel is incompatibility breaking.
All the channels in the double-hatched region are entanglement breaking. 
For the single-hatched region, the color coding gives the probability for the 
channel not to be entanglement breaking. 
There are no entanglement breaking channels in the dark blue non-hatched region,
and there are not incompatibility breaking channels when $\mu>\mu_A$.

\begin{figure}[ht!]
       \includegraphics[width=0.33\textwidth]{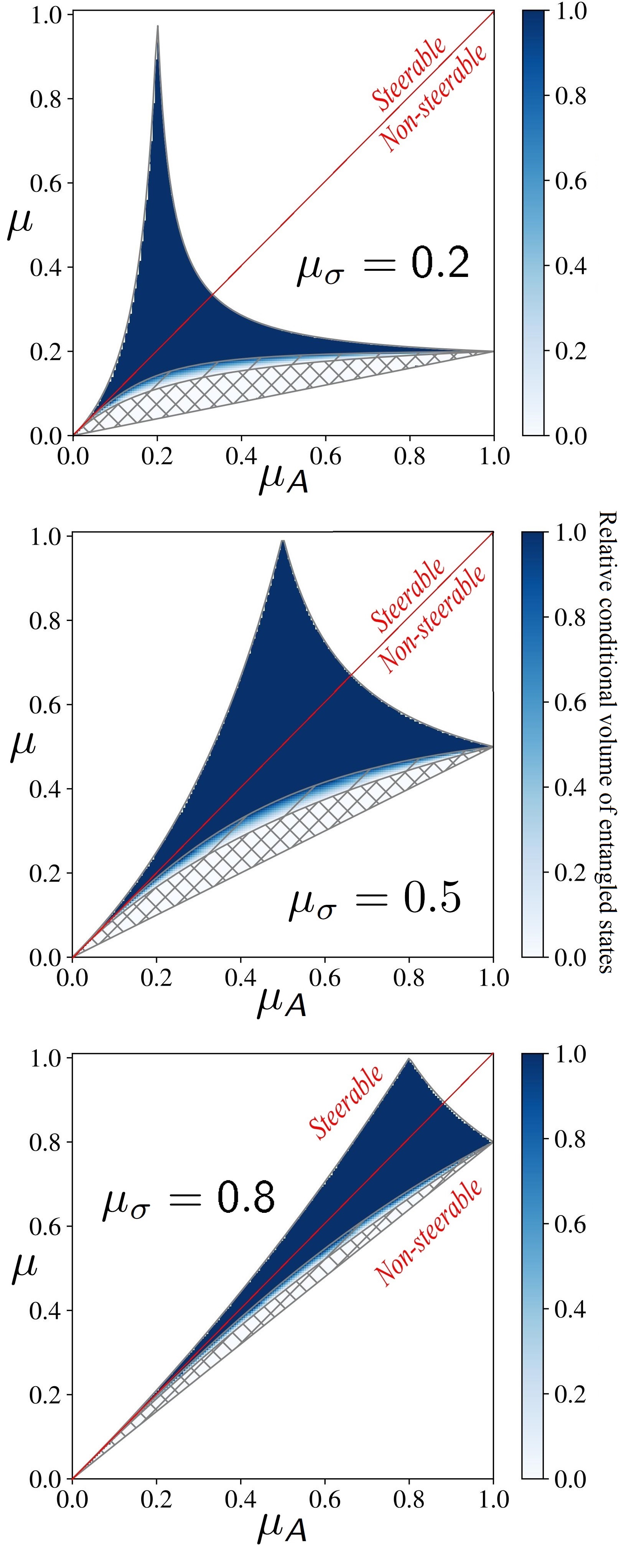}
       \caption{\label{fig:cond_prob}The separability (double-hatched), coexistence (single-hatched), and 
         entanglement (unhatched) regions for the two-mode
         Gaussian CJ states for different marginal purities: $\mu_\sigma=0.2$ (top), $\mu_\sigma=0.5$ (center), and $\mu_\sigma=0.8$ (bottom).
         The red line separates the steerable and non-steerability states. 
         The shading in the coexistence region is  associated with the conditional 
         relative volume of entangled CJ states.  
         In terms of the one-mode channels, the regions 
         correspond to the incompatibility breaking channels (non-steerable CJ states), entanglement
         breaking channels (double-hatched region), and the coexistence of entanglement breaking and not entanglement 
         breaking channels (single-hatched region).}
\end{figure}
%
%
%
%
%
%
%
%
\section{Conclusions and outlook}
In this paper we calculate the relative volumes of entanglement and
incompatibility breaking one-mode Gaussian quantum channels. We use the
Choi-Jamio{\l}kowski isomorphism to define the geometry in the space 
of Gaussian quantum channels. We explicitly determine the Hilbert-Schmidt line
and volume elements for one-mode Gaussian channels, together with
the regions corresponding to completely positive maps, as well as entanglement breaking
and incompatibility breaking channels. Interestingly, these regions 
are completely characterized by inequalities involving only
two channel parameters: $\det M$ and $\det N$. 
We find it useful to express the volume element in terms of symplectic 
invariants.

We base all of our calculations on the Hilbert-Schmidt
distance. It would be interesting to compare our results with the
volumes obtained from  the Fisher-Rao and Bures-Fisher metrics.
We would also like to explore the connection to a geometry of channels 
defined with 
the help of probe states~\cite{Monras2010}.

Geometrical concepts are also relevant for  
discrimination, distinguishability, and tomography of Gaussian channels. 
Having the general framework at hand, 
it is now straightforward to study the geometrical properties of
subclasses of special interest, such as  
Weyl-covariant and quantum limited channels. Another open problem is finding relative volumes for canonical 
classes of channels, as introduced by 
Holevo in \cite{Holevo2007} and further explored
in~\cite{Caruso_2006}.

\section*{Acknowledgments} 
K.S. was supported by the Polish National Science Centre project No. 2018/28/T/ST2/00008.
The authors would like to thank Dariusz Chru{\'s}ci{\'n}ski, Valentin Link, Moritz Richter, and Roope Uola for stimulating discussions.

\appendix
\section{Hilbert-Schmidt line element}\label{sec:hilbert-schmidt-line}
Our computation follows the lines of~\cite{Link2015,Sohr_2018}. From 
our definition, the Hilbert-Schmidt line element reads
\begin{equation}\label{dhs2}
\begin{split}
\der s^2=&\Tr[\rho(\Sigma+\der\Sigma,\ell+\der\ell)]^2
+\Tr[\rho(\Sigma,\ell)]^2\\&
-2\Tr[\rho(\Sigma,\ell)\rho(\Sigma+\der\Sigma,\ell+\der\ell)].
\end{split}
\end{equation}
Observe that
\begin{equation}
\begin{split}
\Tr[\rho(\Sigma,\ell)&\rho(\Sigma^\prime,\ell^\prime)]
=\frac{1}{\sqrt{\det\frac 12 (\Sigma+\Sigma^\prime)}}\\
&\times\exp\left[-\frac 12 (\ell-\ell^\prime)^T(\Sigma+\Sigma^\prime)^{-1}(\ell-\ell^\prime)\right],
\end{split}
\end{equation}
where we used the property of the trace
$
\Tr\left[D(-\xi)D(-\xi^\prime)\right]=\pi^{n}\delta^{2n}(\xi+\xi^\prime)
$
and the $2n$-dimensional Gaussian integral \cite{byron2012mathematics}
\begin{equation}
\begin{split}
\int_{\mathbb{R}^{2n}}\frac{\der^{2n}\xi}{\pi^n}\exp&\left[-\frac 12 \xi^TA\xi+B^T\xi\right]\\&=\frac{2^n}{\sqrt{\det A}}\exp\left[\frac 12 B^TA^{-1}B\right].
\end{split}
\end{equation}
Hence, eq. (\ref{dhs2}) simplifies to
\begin{equation}
\begin{split}
\der s^2=&\frac{1}{\sqrt{\det\Sigma}}
+\frac{1}{\sqrt{\det(\Sigma+\der\Sigma)}}-\frac{2}{\sqrt{\det\frac 12 (2\Sigma+\der\Sigma)}}\\&
\times\exp\left[-\frac 12 \der\ell^T(2\Sigma+\der\Sigma)^{-1}\der\ell\right].
\end{split}
\end{equation}
By expanding the matrix $(2\Sigma+\der\Sigma)^{-1}\simeq\frac 12 (\mathbb{I}-\frac 12 \Sigma^{-1}\der\Sigma)\Sigma^{-1}$ and the exponential
\begin{equation}
\begin{split}
\exp&\left[-\frac 12 \der\ell^T(2\Sigma+\der\Sigma)^{-1}\der\ell\right]\\
&\simeq
1-\frac 12 \der\ell^T(2\Sigma+\der\Sigma)^{-1}\der\ell\\
&\simeq
1-\frac 14 \der\ell^T\Sigma^{-1}\der\ell
\end{split}
\end{equation}
up to the quadratic terms in $\der\Sigma$, $\der\ell$, we find the final formula for the line element.

\section{Hilbert-Schmidt volume element}\label{sec:hilb-schm-volume}
The one-mode Gaussian channels correspond to the two-mode Gaussian CJ states with
\begin{equation}
\Sigma=
\begin{pmatrix}
N+M^T\Sigma_\sigma M & M^TS_\sigma^TZ_\sigma S_\sigma \\
S_\sigma^TZ_\sigma S_\sigma M & \Sigma_\sigma
\end{pmatrix}.
\end{equation}
There always exists a symplectic transformation $S_A$ such that $N+M^T\Sigma_\sigma M=\nu_AS_A^TS_A$. Hence, we can write
\begin{equation}
\Sigma=(S_A\oplus S_\sigma)^T
\begin{pmatrix}
\nu_A\mathbb{I}_2 & S_A^{-1}M^TS_\sigma^TZ_\sigma \\
Z_\sigma S_\sigma MS_A^{-1} & \nu_\sigma\mathbb{I}_2
\end{pmatrix}
(S_A\oplus S_\sigma).
\end{equation}
The off-diagonal block has the singular value decomposition
$S_A^{-1}M^TS_\sigma^TZ_\sigma=Q^T{\Gamma}R$ with two
orthogonal matrices $Q$, $R$ and
${\Gamma}:=\mathrm{diag}(\gamma_+,\gamma_-)$
\cite{Duan2000}. Therefore, one has
\begin{equation}
\Sigma=(S_A^{\prime T}\oplus S_\sigma^TR^T)W(S_A^{\prime}\oplus RS_\sigma)
\end{equation}
with $S_A^\prime:=QS_A$ and
\begin{equation}
W=\begin{pmatrix}
\nu_A\mathbb{I}_2 & {\Gamma} \\
{\Gamma} & \nu_\sigma\mathbb{I}_2
\end{pmatrix}.
\end{equation}
The line element follows from eq. (\ref{line}). We use the fact that
the covariance matrix has the structure $\Sigma=S^TWS$ with symplectic
$S=\exp(H)$ that is generated by a traceless Hamiltonian matrix
$H$. This gives
\begin{equation}
\der\Sigma=S^T(\der W+\der H^TW+W\der H)S,
\end{equation}
and therefore
\begin{equation}
\begin{split}
\der s^2=&
\frac{1}{16\sqrt{\det W}}\Big\{2\Tr[W^{-1}(\der W+\der H^TW+R\der W)]^2\\&
+[\Tr(W^{-1}\der W)]^2\Big\}.
\end{split}
\end{equation}
In the last step, we perform the change of coordinates to the purity-seralian coordinates in eq. (\ref{pur}).

\bibliography{TypicalGaussian2}
\bibliographystyle{beztytulow2}

\end{document}